%% file: ms.tex
\documentclass[a4paper,10pt]{article}
\usepackage[margin=2cm]{geometry}
\usepackage{setspace}
\usepackage{amsmath,amsthm,amsfonts,amssymb}
\usepackage{xcolor}
\usepackage[ruled]{algorithm2e}
\usepackage{cases}
\usepackage{graphicx}
\usepackage{booktabs}
\usepackage{diagbox}
\usepackage{multirow, bigdelim, makecell}
\usepackage{microtype}
\allowdisplaybreaks

\usepackage[textsize=scriptsize,textwidth=3cm]{todonotes}

\title{Knapsack Secretary Through Boosting\thanks{Supported in part by the Independent Research Fund Denmark, Natural Sciences, grant DFF-0135-00018B.}} 

\author{Andreas Abels\thanks{School of Business and Economics, RWTH Aachen University, Germany. Email: \texttt{andreas.abels@oms.rwth-aachen.de}.} \and
Leon Ladewig\thanks{Munich, Germany. Email: \texttt{leonladewig@mail.de}.} \and
Kevin Schewior\thanks{Department of Mathematics and Computer Science, University of Southern Denmark, Odense, Denmark. Email: \texttt{kevs@sdu.dk}.} \and
Moritz Stinzendörfer\thanks{Department of Mathematics, TU Kaiserslautern, Germany. Email: \texttt{stinzendoerfer@mathematik.uni-kl.de}.}}

\newcommand{\OPT}{v(\mathsf{OPT})}
\newcommand{\ALGa}{\mathsf{ALG_\alpha}}
\newcommand{\ALG}{\mathsf{ALG}}
\newcommand{\alg}{\ALG}
\newcommand{\opt}{\mathsf{OPT}}
\newcommand{\E}{\mathbb{E}}
\renewcommand{\epsilon}{\varepsilon}

\newtheorem{theorem}{Theorem}
\newtheorem{lemma}{Lemma}
\newtheorem{definition}{Definition}
\newtheorem{corollary}{Corollary}
\newtheorem{observation}{Observation}
\newtheorem{proposition}{Proposition}

\begin{document}

\maketitle

\input{abstract}

\section{Introduction}
\label{sec:intro}

\input{intro}

\section{Preliminaries}
\label{sec:prelims}

\input{prelims}

\section{Matching $1/e$ for 1-2-Knapsack}
\label{sec:1-2}

\input{algorithm}
\input{no-boosting}
\input{boosting}

\section{Ordinal Algorithms for 1-$B$-Knapsack}
\label{sec:ordinal}

\input{impossible}

\section{Conclusion}
\label{sec:conclusion}
\input{conclusion}
 
\newpage
\bibliographystyle{abbrv}
\bibliography{ks}

\newpage
\appendix
\input{apx}

\end{document}

%% file: abstract.tex
\begin{abstract}
	We revisit the knapsack-secretary problem (Babaioff et al.; APPROX 2007), a generalization of the classic secretary problem in which items have different sizes and multiple items may be selected if their total size does not exceed the capacity $B$ of a knapsack. Previous works show competitive ratios of $1/(10e)$ (Babaioff et al.), $1/8.06$ (Kesselheim et al.; STOC 2014), and $1/6.65$ (Albers, Khan, and Ladewig; APPROX~2019) for the general problem but no definitive answers for the achievable competitive ratio; the best known impossibility remains~$1/e$ as inherited from the classic secretary problem. In an effort to make more qualitative progress, we take an orthogonal approach and give definitive answers for special cases.
	
	Our main result is on the $1$-$2$-knapsack secretary problem, the special case in which $B=2$ and all items have sizes $1$ or $2$, arguably the simplest meaningful generalization of the secretary problem towards the knapsack secretary problem. Our algorithm is simple: It \emph{boosts} the value of size-$1$ items by a factor $\alpha>1$ and then uses the size-oblivious approach by Albers, Khan, and Ladewig. We show by a nontrivial analysis that this algorithm achieves a competitive ratio of $1/e$ if and only if~\mbox{$1.40\lesssim\alpha\leq e/(e-1)\approx 1.58$}.
	
	Towards understanding the general case, we then consider the case when sizes are $1$ and $B$, and $B$ is large. While it remains unclear if $1/e$ can be achieved in that case, we show that algorithms based only on the relative ranks of the item values can achieve precisely a competitive ratio of~$1/(e+1)$. To show the impossibility, we use a non-trivial generalization of the factor-revealing linear program for the secretary problem (Buchbinder, Jain, and Singh; IPCO 2010).
\end{abstract}

%% file: intro.tex
In the classic secretary problem, there is a single position to be filled, and $n$ candidates arrive one by one in uniformly random order. Upon arrival of any candidate, they have to be rejected or accepted immediately and irrevocably only based on \emph{ordinal} information on the candidates seen so far, that is, their relative ranks. The goal is to maximize the probability that the best candidate is selected. The origin of this problem is unclear; for a discussion, we refer to Ferguson's survey~\cite{Ferg89a}. It is well known~\cite{Lindley1961,Dynkin1963} since the 1960s that a probability of $1/e$ can be achieved by selecting the first candidate that is better than the~$n/e$ first candidates and that this is the best-possible probability under the typical assumption $n\rightarrow\infty$. Many extensions of this problem have since been considered, especially in recent years, partially due to relations to beyond-the-worst-case analyses of online algorithms (e.g.,~\cite{Kenyon96,Albers0L20,gupta_singla_2021}) and to mechanism design (e.g.,~\mbox{\cite{Kleinberg05,BabaioffIKK07}}).

There is extensive work on multiple-choice variants of the secretary problem. Few of these works consider an ordinal setting~\cite{BuchbinderJS14,HoeferK17,SotoTV18}; the majority considers the \emph{value} setting in which each arriving candidate (or item) $i$ is revealed along with a value $v_i\in\mathbb{R}_{\geq 0}$ and must be rejected or accepted immediately and irrevocably so that the set of accepted items obeys some combinatorial constraint. The goal is to obtain an algorithm with a (strong) competitive ratio $\rho$, i.e., that constructs a solution $\mathsf{ALG}$ such that $v(\mathsf{ALG})$, the sum of values of accepted items, is in expectation at least $\rho\cdot v(\mathsf{OPT})$ where $\mathsf{OPT}$ is the best solution that could have been constructed.

Whereas the results for the standard secretary problem carry over to the value setting, even relatively simple variants are not completely understood in that setting. This is arguably due to the sheer amount of conceivable strategies. For instance, the precise competitive ratio achievable in the $2$-secretary problem, the variant in which two positions are to be filled, is \emph{not} known---only that it is strictly larger than in the much-better-understood ordinal ``counterpart'', sometimes called the $(2,2)$-secretary problem~\cite{BuchbinderJS14,ChanCJ15}.

The secretary variant that has probably received most attention is the matroid secretary problem~\cite{BabaioffIKK18}, an extension of the $k$-secretary problem~\cite{Kleinberg05} (in which~$k$ positions are to be filled) to any matroid constraint, see, e.g., the state-of-the-art result~\cite{Lachish14,FeldmanSZ18} and the survey by Dinitz~\cite{Dinitz13}. An orthogonal and also well-known extension of $k$-secretary is the \emph{knapsack} secretary problem in which items additionally have sizes and the total size of accepted items must not exceed some given capacity $B$~\cite{BabaioffIKK07,KesselheimRTV18,AlbersKL21,NaoriR19,KesselheimM20}. While this line of work has improved the competitive ratio from $1/(10e)$ to $1/6.65$, no impossibility beyond $1/e$ has been found. For some secretary versions, e.g., the bipartite-matching variant~\cite{KesselheimRTV13}, it is known that this ratio can in fact be matched.

Our paper may raise hope that the ratio of $1/e$ can in fact be matched for knapsack secretary. First, we consider the $1$-$2$-knapsack problem. Here, items have sizes either $1$ or $2$ and the capacity $B$ is $2$. We develop a $1/e$-competitive algorithm. To us, this result is both surprising and significant because the problem generalizes both the classic secretary problem, which severely restricts the set of candidate algorithms, and the not-entirely-understood $2$-secretary problem. 
We also consider the problem with sizes either $1$ or $B$ and $B$ large, for which we show initial results, namely that $1/(e+1)\pm o(1)$ is precisely the competitive ratio that can be achieved by \emph{ordinal} algorithms. These are algorithms that only use the relative rank of the items and disregard the actual values.

\subsection{Related Work}

Kleinberg~\cite{Kleinberg05} first considers $k$-secretary as introduced above, gives an algorithm with competitive ratio $1-\Theta(1/\sqrt{k})$, and shows that this ratio is asymptotically best possible. This result is reproduced by Kesselheim et al.~\cite{KesselheimRTV18} in the more general context of packing LPs. Buchbinder et al.~\cite{BuchbinderJS14} consider the~$(j,k)$-secretary problem in the ordinal setting in which $j$ items can be selected and the goal is to maximize the expected ratio of elements selected from the top $k$ items. They also state the algorithm-design problems as linear programs, which they can only solve for small values of $j$ and $k$, but Chan et al.~\cite{ChanCJ15} can solve them for larger values. Any guarantee for the $(k,k)$-secretary problem carries over to the $k$-secretary problem, but Chan et al.~\cite{ChanCJ15} rule out the other direction. More specifically, Chan et al.'s results include an optimal algorithm for $(2,2)$-secretary with guarantee approximately $0.489$ and a (not necessarily optimal) algorithm for $2$-secretary with guarantee approximately $0.492$. Albers and Ladewig~\cite{AlbersL21} revisit the problem and give simple algorithms with improved (albeit non-optimal) competitive ratios for many fixed values of~$k$.

The knapsack secretary problem is introduced by Babaioff et al.~\cite{BabaioffIKK07} who give a~$1/(10e)$-competitive algorithm, which was subsequently improved by Kesselheim et al.~\cite{KesselheimRTV18} to $1/8.06$ and by Albers, Khan, and Ladewig~\cite{AlbersKL21} to $1/6.65$. Essentially all known $\Omega(1)$-competitive algorithms for the knapsack secretary problem are somewhat wasteful in the competitive ratio, presumably at least partially for the sake of a simpler analysis, in that they randomize between different algorithms that are tailored to respective item sizes. It seems that qualitative progress can only be made by a more fine-grained analysis avoiding such case distinctions.

A variant of the knapsack secretary problem that has recently been considered is the fractional variant in which an item can also be packed fractionally, avoiding situations in which an arriving item cannot be selected at all, even when there is space. The currently best known achievable competitive ratio is $1/4.39$~\cite{GilibertiK21}, also achieved by a blended approach.

It is not difficult to see that no constant competitive ratio can be achieved when the items do not arrive in random but in adversarial order, even in the unit-value case~\cite{Marchetti-SpaccamelaV95}. Starting from this problem, problems in which other assumptions than the order are relaxed are considered as well. For instance, Zhou et al.~\cite{ZhouCL08} consider the version in which each item has a small size; Böckenhauer et al.~\cite{BockenhauerKKR14} and Boyar et al.~\cite{BoyarFL22} introduce advice and untrusted predictions, respectively, to the problem.

Lower bounds for secretary problems in the value setting are rare. For some related problems~\cite{CorreaDFS19,CorreaDFSZ21,abs-2011-01559}, the rich class of strategies can be handled by, for any strategy, identifying an infinite set of values (using Ramsey theory) on which it is much better behaved. It is, however, not clear how such an approach could be applied, e.g., for knapsack secretary since it seems one would need to control how the values in the support are spread out, a property that is irrelevant in the other settings.

\subsection{Our Contribution}

The special case $1$-$2$-knapsack is not only arguably the simplest special case that exhibits features of the knapsack problem distinguishing it from the matroid secretary problem. Since the problem generalizes both the standard secretary problem and $2$-secretary, we believe that settling it in terms of the achievable competitive ratio is also interesting per se. 

A good starting point for tackling $1$-$2$-knapsack seems to be the extended secretary algorithm, which is $1/3.08$-competitive in the slightly more general case when all items have size larger than $B/3$~\cite{AlbersKL21}. This algorithm simply ignores the item sizes, samples some prefix of length $cn$ for some optimized constant~$c\in(0,1)$, and afterwards selects all items that surpass the largest value from the sampling phase and that can still be feasibly packed. It is, however, easy to see that this approach cannot achieve $1/e$: Achieving $1/e$ in an instance where the optimal solution consists of a large item requires setting $c=1/e\pm o(1)$. The resulting algorithm will, however, not be $1/e$-competitive in an instance where the optimal solution consists of two small items of equal value, but there are many large items, each slightly more valuable than the individual small items, making sure that the small items are (almost) never selected by the algorithm. In this case, the competitive ratio of the algorithm will be essentially half the probability that the algorithm selects a (large) item, that is, $(1-1/e)/2<1/e$. We denote two instances of the above forms by $\mathcal{I}_1$ and $\mathcal{I}_2$, respectively, in the following. Clearly, it is possible to choose $c$ so as to balance between $\mathcal{I}_1$ and $\mathcal{I}_2$. As a small side result, we show that a ratio of approximately $0.353<1/e$ can be achieved that way.

The key observation leading to our $1/e$-competitive algorithm is that keeping~$c=1/e$ and internally multiplying (\emph{boosting}) values of small items with a suitable constant factor $\alpha>1$ prior to running the extended secretary algorithm may handle both $\mathcal{I}_1$ and $\mathcal{I}_2$: While this is clear for $\mathcal{I}_1$ when the ranking of values does not change through boosting, a small item may overtake the most valuable (large) item. This however means that this small item has relatively large (actual) value. Using that the algorithm also accepts the second-best item with a significant probability ($1/e^2$), we can show that, with the right choice of $\alpha$, we still extract enough value from the small and large items to cover $1/e\cdot \OPT$. In $\mathcal{I}_2$, the small items would overtake the large items, significantly improving the expected value achieved by the algorithm; conversely, if they did not overtake, they would not have been harmfully valuable in the first place---again with the right choice of $\alpha$. To sum up, ``$\mathcal{I}_1$ type'' instances impose an upper bound on $\alpha$, and ``$\mathcal{I}_2$ type'' instances impose a lower bound on $\alpha$. We show that the algorithm is~$1/e$-competitive if \emph{and only if} $1.40\lesssim\alpha\leq e/(e-1)\approx 1.58$ where the upper bound comes essentially from the above consideration for $\mathcal{I}_1$. Note that therefore, in particular, our boosting \emph{is} different from ordering the items by their ``bang for the buck'' ratios.

We note that, while $\alpha$-boosting seems reminiscent of $\beta$-filtering~\cite{ChanCJ15} (for $\beta<1$), applying $\beta$-filtering to the extended secretary algorithm will not yield a $1/e$-competitive algorithm. The extended secretary algorithm would be adapted by ignoring items with a value less than $\beta$ times the highest value seen so far. Note that indeed, a ``$\mathcal{I}_1$ type'' instance where all but the most valuable item have a similar small value, one would have to choose $c=1/e\pm o(1)$ again, independently of $\beta$. But such an algorithm would again only be $(1-1/e)/2$-competitive on $\mathcal{I}_2$.

The crux of our analysis is distinguishing all possible cases beyond those covered by $\mathcal{I}_1$ and $\mathcal{I}_2$ in a smart way. To bound the algorithm's value in each of these cases, we precisely characterize the probabilities with which the algorithm selects an item depending on its size and its position in the (boosted) order of values, significantly extending observations made by Albers and Ladewig~\cite{AlbersL21}.

Before tackling the general case and understanding potentially complicated knapsack configurations, we propose considering a clean special case called \mbox{$1$-$B$}-knapsack where items have sizes either $1$ or $B$, and $B$ is large. One may be tempted to think that this special case is difficult in that selecting a small item early on may lead to a blocked knapsack and a horribly inefficient use of capacity, e.g., because all other items are large. On the other hand, when $B$ is large, one can easily avoid such situations by sampling. We do not give a conclusive answer on whether $1/e$ can be matched in this case, but we give some preliminary results.

Unfortunately, a competitive ratio of $1/e$ for $1$-$B$-knapsack cannot be achieved with our boosting approach. The same consideration we made for $\mathcal{I}_1$ earlier (for~\mbox{$1$-$2$}-knapsack) to get an upper bound of $e/(e-1)$ on $\alpha$ still works; in contrast, a generalization of $\mathcal{I}_2$ rules out any constant boosting factor. 

We then give another algorithm for $1$-$B$ knapsack which can be viewed as a linear interpolation between the classic secretary algorithm and the algorithm by Kleinberg~\cite{Kleinberg05} for $k$-secretary. We show that it is~$1/(e+1)$-competitive. This algorithm turns out to be \emph{ordinal}, that is, its decisions only depend on the item sizes and the relative order of their values. Remarkably, we are able to show that~$1/(e+1)$ is the best-possible guarantee such algorithms can achieve. We do so by generalizing the factor-revealing linear program due to Buchbinder et al.~\cite{BuchbinderJS14} by adding variables and constraints. Arguing that the LP indeed models our problem becomes more difficult because, in contrast to the setting of Buchbinder et al., at any time, even the size of the next item is random. We do so by showing reductions between our model and an auxiliary batched-arrival model.

%% file: prelims.tex
We use the following notation.
Let $\mathcal{I} = \{1,\ldots,n\}$ be the set of items (also called \textit{elements}), where each item $i \in \mathcal{I}$ is specified by a profit $v_i$ and a size $s_i$. Moreover, we are given a knapsack of capacity $ B \in \mathbb{N}_{\geq 2}$.
The goal is to find a maximum-profit packing, i.e., a subset of items $S$ such that~$\sum_{i \in S} s_i \leq B$ and~$\sum_{i \in S} v_i$ is maximized.
Without loss of generality, we assume that all elements have distinct values and that $v_1 > v_2 > ... > v_n$. This way, the name of an item $i$ corresponds to the (global) \textit{rank} in $\mathcal{I}$.

Throughout the following sections, an important subclass of the knapsack problem arises where each item has either size 1 or $B$.
\begin{definition}[1-$B$-knapsack]
	We call the special case of the knapsack problem where all items have size $1$ or $B$ the \textit{1-$B$-knapsack problem}. Items of size $1$ are called \textit{small} and items of size $B$ are called \textit{large}.
\end{definition}
Within the context of 1-$B$-knapsack, we use the following further notation.
Let~$\mathcal{I}_S$ be the set of small items. For any small item $i \in \mathcal{I}_S$, let $\mathrm{r}_\mathrm{s} (i)$ denote its rank among the small items.
Note that $\mathrm{r}_\mathrm{s} (i)$ is at most the global rank $i$ of this item.
Further, let $ \mathrm{r}'_\mathrm{g}(a) $ denote the global rank of the small item $x$ that satisfies~$\mathrm{r}_\mathrm{s}(x) = a$.
When we use just the word ``rank'', we refer to the global rank.

Let $\opt$ be an optimal offline algorithm. For any algorithm $\ALG$, we overload the notation and use the same symbol also for the packing returned by the algorithm. Further, we denote by $v(\ALG) := \sum_{i \in \ALG} v_i$ the total profit of the packing returned by $\ALG$. We are particularly interested in \emph{online} algorithms, i.e., algorithms that are initially only given $n$ and are presented with the items one by one. Upon arrival of any item, an online algorithm has to irrevocably decide whether it includes the item or not. A special class of algorithms we consider are \emph{ordinal} algorithms. These algorithms only have access to the item sizes and the \emph{relative order} of item values.

We say that an online algorithm $\ALG$ is $\rho$-competitive if $\mathbb{E}[v(\ALG)]\geq\rho\cdot \OPT$ for all instances, where the expectation is taken over a uniformly random arrival order (and possibly internal randomization that the algorithm uses).
In general, we assume $n\rightarrow\infty$ for our bounds. Note that, for a fixed number of items, we can achieve a guarantee that is arbitrarily close to the guarantee for $n\rightarrow\infty$ by adding a sufficient amount of virtual dummy items.

Finally, throughout the paper, we use the notation $[k]:=\{1,\dots,k\}$ for any~$k\in\mathbb{N}$.

%% file: algorithm.tex
In this section, we develop an optimal algorithm for 1-2-knapsack.
For this purpose, we first propose a natural algorithm for 1-$B$-knapsack, based on the size-oblivious approach from~\cite{AlbersKL21}. Here, items are accepted whenever their profit exceeds a certain threshold, similar to the optimal algorithm for the classic secretary problem. Therefore, we call it the \textit{extended secretary algorithm}. 
From an initial sampling phase of length $cn$, where $c \in (0,1)$ is a parameter of the algorithm, the best item is used as a reference element. Subsequently, any item beating the reference element is packed if it still fits.
A formal description is given in Algorithm~\ref{alg:extendedSec}.

\begin{algorithm}[t]
	\textbf{Input:} Instance of 1-$B$-knapsack arriving in uniformly random order, parameter $ c \in (0,1) $.\\
	\textbf{Output:} A knapsack packing.\\
	\For{round $ \ell =1 $ to $ n $ }{
		\If{$ \ell \leq cn$}{ 
			Reject the current item\tcp*{sampling phase}
		}  
		\If{$\ell > cn$}{
			Let $ v^* $ be the highest profit seen up to round $ \lfloor cn \rfloor$\;
			Pack the current item if its profit exceeds $ v^* $ and the remaining capacity is large enough\;
		}
	}
	\caption{Extended secretary algorithm}
	\label{alg:extendedSec}
\end{algorithm}

In the following, we denote Algorithm~\ref{alg:extendedSec} by $\ALG$ and set
\begin{equation*}
\begin{array}{ccl}
p_i(j) &:=& \text{Pr}[\mathsf{ALG} \text{ packs item $i$ as the $j$-th element}],\\
p_i &:=& p_i(1),\\
P_i &:=& \sum\limits_{j=1}^B p_i(j).
\end{array}
\end{equation*}
Thus, $P_i$ is the probability that the algorithm packs item $i$ at all, while $p_i$ is the probability that it is packed as the first item. 
We first state some results on the values $p_i$, which have essentially been investigated in~\cite{AlbersL21}. Indeed, the following results follow from that work and some simple observations.

\begin{lemma}\label{wkeitenc}
	For $ i \in \mathbb{N} $, it holds that
	\begin{equation*}
	p_i = c \left( \ln \dfrac{1}{c} + \displaystyle\sum\limits_{ \ell = 1}^{i-1} (-1)^{\ell +1} \binom{i-1}{\ell} \dfrac{c^{\ell} -1}{\ell} \right)\pm o(1).
	\end{equation*}
\end{lemma}
\begin{proof}
	Let $i \in \mathbb{N}$. The extended secretary algorithm packs $i$ as the first item if and only if the \textsc{single-ref} algorithm from \cite{AlbersL21} with $r=1$ and $k=i$ packs $i$ as the first item. 
	Hence, the probability $p_i$ can be derived from \cite{AlbersL21} as follows:
	If~$i=1$, item $i$ is a dominating item in the terminology of \cite{AlbersL21} and Lemma~6 of~\cite{AlbersL21} gives $p_1 = c \cdot \ln(1/c) - o(1)$.
	In the case $i \geq 2$, item $i$ is a non-dominating item in the terminology of \cite{AlbersL21}. Here, Lemma 4 of \cite{AlbersL21} gives $p_i= p_i(i)$ and Lemma~5 of~\cite{AlbersL21} and gives $p_i(i) = p_1(i)$, that is, $p_i$ turns out to be the probability that the dominating item $1$ is accepted as the $i$-th item by the \textsc{single-ref} algorithm. Again, the claim follows from Lemma~6 of~\cite{AlbersL21}.
\end{proof}

Furthermore, observe that, since increasing the profit of an item cannot decrease its probability of being selected, we have $p_i \geq p_{i+1}$ for all $i \in [n-1]$. Note that $\mathsf{ALG}$ accepts no item if and only if the best item is in the sampling phase. Therefore, we have the following observation.

\begin{observation}\label{obs:sumOfPi}
	It holds that
\begin{align*}
\sum_{i=1}^n p_i &= 1 - \Pr[\text{\upshape $\mathsf{ALG}$ accepts no item}] \\ &= 1 - \Pr[\text{\upshape item $1$ appears in sampling phase}] = 1 - c \,.
\end{align*}
\end{observation}
In the following subsection, we identify relations between the probabilities $P_i$ and $p_j$.

\subsection{Structural Lemma}
\label{sec:extendedSecProperties}

In this subsection, we show the following lemma connecting the probabilities $P_i$ to the probabilities $p_j$ from Lemma~\ref{wkeitenc}. The analysis showing the $1/e$-competitiveness of our algorithm 
is crucially based on this result. Note that we only use it for $B=2$ but it holds for all $B$.
\begin{lemma}\label{lemmaWkeit}
	The probability that $\mathsf{ALG}$ packs element $i \in \mathcal{I}$ is
	\begin{numcases}
	{P_i=}
	p_i & \text{if element $i$ is large,} \label{eq:l1large} \\
	i_s^* \cdot p_i + \sum\limits_{x=\mathrm{r}_\mathrm{s}(i) +1}^{B^*} p_{\mathrm{r}'_\mathrm{g}(x)} & \text{if element $i$ is small}, \label{eq:l1small}
	\end{numcases}
	with $ i_s^* := \min\{\mathrm{r}_\mathrm{s}(i), B\}  $ and $ B^* := \min\left\{ B, |\mathcal{I}_S| \right\}$.
\end{lemma}
Observe that \eqref{eq:l1large} follows immediately: Any large element can only be packed when the knapsack is empty, i.e., as the first element. The proof of \eqref{eq:l1small} requires a bit more work.
\begin{definition}
	Let $ E_{x,y}^{i,j} $ be the event that the small elements $ i $ and $j$ are packed as the $x$-th and $y$-th items, respectively.
\end{definition}

Note that the event that any item $i \in \mathcal{I}_S$ is packed as $x$-th item, where $x \geq 2$, can be partitioned according to the item packed first. Therefore, for any $i \in \mathcal{I}_S$ and $x \geq 2$,
\begin{equation}\label{HL0}
p_i(x) = \sum\limits_{j \in \mathcal{I}_S} \normalfont \text{Pr} \left[E_{1,x}^{j,i}\right] \,.
\end{equation}
We have the following technical lemmata.

\begin{lemma}\label{HL1}
	Let $i \in \mathcal{I}_S$ be any small item and $i_s^* = \min\{\mathrm{r}_\mathrm{s}(i), B\}$. 
	For $ 2 \leq \ell \leq i_s^* $, it holds that
	$ \displaystyle\sum\limits_{j \in \mathcal{I}_S} \normalfont \text{Pr}\left[E_{1,\ell}^{i,j}\right] = p_i $.
\end{lemma}
\begin{proof}
	The first step is to show that at least $\ell$ elements are accepted in total, if element $i$ is accepted first. Since element $ i $ has rank $\mathrm{r}_\mathrm{s}(i) $ among the small elements, there are $\mathrm{r}_\mathrm{s}(i) -1$ small elements that are more valuable.
	Their position in the input sequence cannot be in the sampling phase, nor before element $ i$ if it is packed first. So there are at least $i_s^*$ small elements that can be packed subsequently.
	Therefore, for $ 2 \leq \ell \leq i_s^* $, a small element is packed as $\ell$-th item.
	The claim follows by partitioning the event that $i$ is packed first according to the item $j \in \mathcal{I}_S$ packed as $\ell$-th item.
\end{proof}

\begin{lemma}\label{HL2}
	For any two small elements $i, j \in \mathcal{I}_S$ and any $x,y \in [B]$, we have
	$ \normalfont \text{Pr}\left[E_{x,y}^{i,j}\right] = \normalfont \text{Pr}\left[E_{x,y}^{j,i}\right]  $.
\end{lemma}
\begin{proof}
	Consider any input sequence of $ E_{x,y}^{i,j} $ and the sequence resulting from swapping the elements $ i $ and $ j $. Since both elements are not part of the sample, the reference element is not changed by the swap. Therefore, no element that was previously accepted will be rejected and none that was previously rejected will be accepted. Only the order of selection changes.
\end{proof}
\begin{lemma}\label{HL3}
	For any small items $i,j,m \in \mathcal{I}_S$ with 
	$ \mathrm{r}_\mathrm{s}(m) > 1$
	and $ \mathrm{r}_\mathrm{s}(i) < \mathrm{r}_\mathrm{s}(m)$,
	it holds that
	$ \normalfont \text{Pr}\left[E_{1,\mathrm{r}_\mathrm{s}(m)}^{m,j}\right] = \normalfont \text{Pr}\left[E_{1,\mathrm{r}_\mathrm{s}(m)}^{i,j}\right]  $.
\end{lemma}
\begin{proof}
	Consider any input sequence from $E_{1,\mathrm{r}_\mathrm{s}(m)}^{m,j} $. Since $ \mathrm{r}_\mathrm{s}(i) < \mathrm{r}_\mathrm{s}(m) $ applies, element $ i $ lies behind the element with rank $ m $ in the sequence. If both are selected (see $ i_1 $ in Figure \ref{fig:fig3}), this also applies after they have been swapped (see Lemma~\ref{HL2} and $ m_1 $ in Figure \ref{fig:fig3}). If previously only element $ m $ of the two is packed, only element $ i $ (of the two) is selected after their swapping ($ i_2 $ and $ m_2 $ in Figure \ref{fig:fig3}), since in this case, nothing changes in the reference element either. Therefore $ \normalfont \text{Pr}\left[E_{1,\mathrm{r}_\mathrm{s}(m)}^{m,j}\right] \leq \normalfont \text{Pr}\left[E_{1,\mathrm{r}_\mathrm{s}(m)}^{i,j}\right]  $ applies.

	Now consider any input sequence from $ E_{1,\mathrm{r}_\mathrm{s}(m)}^{i,j} $. 
	We show that the element $ m $ lies behind the element $ i $ in the sequence since an element is packed as $z$-th item, where $z= \mathrm{r}_\mathrm{s}(m)$.
	Assuming this did not apply and $ m $ is in the sample, then there would be at most $ \mathrm{r}_\mathrm{s}(m)-1 $ small elements that can be packed.
	
	In the case that it occurs in the sequence after the sampling phase, but before element $ i $, there must be a more valuable element in the sample (because $ m $ was not packed) and therefore there are again at most $ \mathrm{r}_\mathrm{s}(m) -1 $ small elements that can be selected.
	In particular, in both cases, no element is packed as $z$-th item for $ \mathrm{r}_\mathrm{s}(m) $. This is a contradiction to the fact that we consider an input sequence in $ E_{1,\mathrm{r}_\mathrm{s}(m)}^{i,j} $.
	Now, using the same argumentation as in the first case, it follows that $ \normalfont \text{Pr}\left[E_{1,\mathrm{r}_\mathrm{s}(m)}^{m,j}\right] \geq \normalfont \text{Pr}\left[E_{1,\mathrm{r}_\mathrm{s}(m)}^{i,j}\right]  $, which completes the proof.
\end{proof}
\begin{figure}[t]
	\centering
	\includegraphics[width=0.8\textwidth]{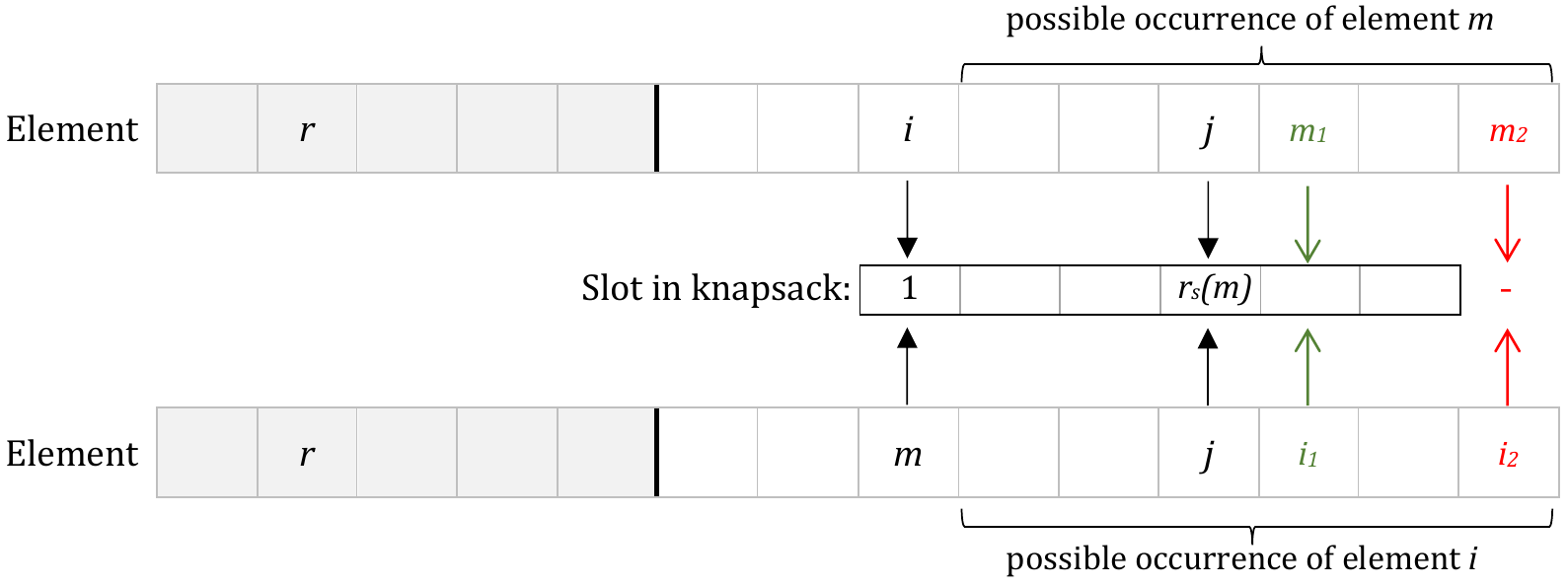} 
	\caption{Occurrence of element $ i $ and $ m $ in event $ E_{1,\mathrm{r}_\mathrm{s}(m)}^{m,j} $ and $ E_{1,\mathrm{r}_\mathrm{s}(m)}^{i,j} $}
	\label{fig:fig3}
\end{figure}

Using Lemmas~\ref{HL1} to~\ref{HL3}, we are now able to prove Lemma~\ref{lemmaWkeit}.
\begin{proof}[Proof of Lemma~\ref{lemmaWkeit}]
	Let $i \in \mathcal{I}$ be any item.
	If $i$ is large, it can only be packed as the first item, thus $ P_i = p_i $.
	Now, assume that $i$ is small. It holds that
	\begin{equation*}
	P_i 
	= \sum\limits_{x=1}^{B^*} p_i(x) 
	= \underbrace{\sum\limits_{x=1}^{i_s^*} p_i(x)}_{(*)} + \underbrace{ \sum\limits_{x=\mathrm{r}_\mathrm{s}(i) +1}^{B^*} p_i(x) }_{(**)}
	\,.
	\end{equation*}
	We next simplify both starred terms using Lemmas~\ref{HL1} to~\ref{HL3}.
	For $ (*) $, it holds that
	\begin{align*}
	\sum\limits_{x=1}^{i_s^*} p_i(x) &=  p_i(1) + \sum\limits_{x=2}^{i_s^*} \sum\limits_{j \in \mathcal{I}_S}  \normalfont \text{Pr}\left[E_{1,x}^{j,i}\right] && \text{(Equation~\eqref{HL0})} \\
	&= p_i + \sum\limits_{x=2}^{i_s^*} \sum\limits_{j \in \mathcal{I}_S}  \normalfont \text{Pr}\left[E_{1,x}^{i,j}\right] && (\text{Lemma \ref{HL2}}) \\
	&=  p_i + \sum\limits_{x=2}^{i_s^*} p_i(1)  && (\text{Lemma \ref{HL1}}) \\
	&= i_s^* \cdot p_i. &&
	\end{align*}
	For $ (**) $, we obtain
	\begin{align*}
	\sum\limits_{x=\mathrm{r}_\mathrm{s}(i) +1}^{B^*} p_i(x) &= \sum\limits_{x=\mathrm{r}_\mathrm{s}(i) +1}^{B^*} \sum\limits_{j \in \mathcal{I}_S}  \normalfont \text{Pr}\left[E_{1,x}^{j,i}\right] && \text{(Equation~\eqref{HL0})} \\
	&= \sum\limits_{x=\mathrm{r}_\mathrm{s}(i) +1}^{B^*} \sum\limits_{j \in \mathcal{I}_S}  \normalfont \text{Pr}\left[E_{1,x}^{i,j}\right] && (\text{Lemma \ref{HL2}}) \\
	&= \sum\limits_{x=\mathrm{r}_\mathrm{s}(i) +1}^{B^*} \sum\limits_{j \in \mathcal{I}_S}  \normalfont \text{Pr}\left[E_{1,x}^{\mathrm{r}'_\mathrm{g}(x),j}\right] && (\text{Lemma \ref{HL3}, }\mathrm{r}_\mathrm{s}(i) < x) \\
	&= \sum\limits_{x=\mathrm{r}_\mathrm{s}(i) +1}^{B^*} p_{\mathrm{r}'_\mathrm{g}(x),} && (\text{Lemma \ref{HL1}, }2 \leq x = \mathrm{r}_\mathrm{s}(\mathrm{r}'_\mathrm{g}(x)) \leq B)
	\end{align*}
	which completes the proof.
\end{proof}

The following corollary is an immediate consequence of Lemma~\ref{lemmaWkeit} for $B=2$.

\begin{corollary}\label{Wkeitencorollary}
	For $B=2$, the probability that $\mathsf{ALG}$ packs element $ i \in \mathcal{I} $ is
	\begin{equation*}
	P_i=
	\begin{cases}{}
	p_i & \text{ if $i$ is large,} \\
	p_i +p_{\mathrm{r}'_\mathrm{g}(2)} & \text{ if $i$ is small and $\mathrm{r}_\mathrm{s}(i)=1$,} \\
	2 p_i & \text{ if $i$ is small and $\mathrm{r}_\mathrm{s}(i)>1$} \,,
	\end{cases}
	\end{equation*}
	where, if the second most valuable small item does not exist, we set $p_{\mathrm{r}'_\mathrm{g}(2)}=0$.
\end{corollary}

%% file: no-boosting.tex
\subsection{First approach: Without Boosting}
\label{sec:12KSwithoutBoosting}
In this subsection, we study Algorithm~\ref{alg:extendedSec} (as is) for 1-2-knapsack. Unfortunately, there are two instances such that it is impossible to choose the parameter $c$ so that Algorithm~\ref{alg:extendedSec} is $(1/e)$-competitive on both instances.

\begin{lemma}
	\label{lemma:notOptimalWithoutBoosting}
	For 1-2-knapsack, the competitive ratio of $\mathsf{ALG}$ is at most $0.35767$,
	assuming $n \to \infty$.
\end{lemma}
\begin{proof}
	Let $ 1 > \epsilon > 0 $ be a constant. We define two instances $\mathcal{I}_1$ and $\mathcal{I}_2$.
	In the first instance $\mathcal{I}_1$, all items are large and only one item has substantial profit. Formally, let $v_1 = 1$, $v_i = \epsilon^i$ for $2 \leq i \leq n$, and $s_i = 2$ for all $1 \leq i \leq n$. Then, for instance $\mathcal{I}_1$,
	\begin{equation}
	\label{eq:extSecAlgUBI1}
	\lim_{\epsilon \to 0}~
	\mathbb{E}[v(\mathsf{ALG})] = P_1 \cdot v_1 = p_1 \cdot v(\mathsf{OPT}).
	\end{equation}
	
	In the second instance $\mathcal{I}_2$, most items are large and essentially of the same profit. However, the optimal packing contains two small items that appear at ranks $n-1$ and $n$. Formally, set $s_i = 2$ for $1 \leq i \leq n-2$, $s_{n-1} = s_{n} = 1$, and~$ v_i = 1 + \epsilon ^i $ for all $ i \in \{1,\ldots,n\} $.
	As item $n$ never beats any reference item, we have $P_n = 0$. Hence, the algorithm selects only items from $\{1,\ldots,n-1\}$ with positive probability, and always at most one item. For instance $\mathcal{I}_2$, we get
	\begin{align}
	&\quad\lim_{\epsilon \to 0}~
	\mathbb{E}[v(\ALG)] 
	= 	\lim_{\epsilon \to 0}~ \sum_{i=1}^{n} (P_i \cdot (1+\epsilon^i)) 
	= \sum_{i=1}^{n} p_i\nonumber\\
	\overset{\text{Obs.~(\ref{obs:sumOfPi})}}{=}&\quad 1 - c
	\leq \frac{1-c}{2} \cdot v(\mathsf{OPT}).\label{eq:extSecAlgUBI2}
	\end{align}
	
	Overall, by Equations~\eqref{eq:extSecAlgUBI1} and~\eqref{eq:extSecAlgUBI2}, the competitive ratio as $n\to\infty$ of $\mathsf{ALG}$ is bounded from above by
	\begin{equation*}\label{eq:c}
	\max\limits_{c \in (0,1) } \min \left\{ p_1,\dfrac{1-c}{2} \right\} = \max\limits_{c \in (0,1) } \min \left\{ c \cdot \ln \dfrac{1}{c},\dfrac{1-c}{2} \right\} \leq 0.35767 \,.
	\end{equation*}
	This completes the proof.
\end{proof}

As a small side result, we show that this bound is almost tight. The techniques are similar to those used for our main result and presented in the full version of the paper.

\begin{proposition}\label{prop:no-boost}
	For $1$-$2$-knapsack, the competitive ratio of $\mathsf{ALG}$ is $ 0.35317-o(1)$,
	setting $c=0.26888$ and assuming $n\rightarrow\infty$.
\end{proposition}

%% file: boosting.tex
\subsection{Optimal algorithm through $\alpha$-Boosting}
\label{sec:12KSwithBoosting}

The proof of Lemma~\ref{lemma:notOptimalWithoutBoosting} reveals the bottleneck of Algorithm~\ref{alg:extendedSec}: If the optimal solution consists of two elements having a high rank, the probability of selecting those items is small. This problem can be resolved by the concept of \textit{$\alpha$-boosting}.

\begin{definition}[$\alpha$-boosting]
	Let $\alpha \geq 1$ be the \textit{boosting factor}.
	For any item $i \in \mathcal{I}$, we define its \textit{boosted profit} to be
	\[
	v'_i = \begin{cases}
	\alpha \cdot v_i & \text{if $i$ is small,} \\
	v_i & \text{otherwise.} \\
	\end{cases}
	\]
\end{definition}

In the following, we investigate Algorithm~\ref{alg:extendedSec} enhanced by the concept of $\alpha$-boosting, denoted by $\mathsf{ALG}_\alpha$. This algorithm works exactly as given in the description of Algorithm~\ref{alg:extendedSec}, but works with the boosted profit $v'_i$ instead of the actual profit $v_i$ for any item $i \in \mathcal{I}$.
Note that the unboosted algorithm analyzed in Proposition~\ref{prop:no-boost} is $\mathsf{ALG}_1$.
For the remainder of this subsection, we fix $c = 1/e$. In particular, this implies $p_1 = 1/e\pm o(1)$ and $p_2 = 1/e^2\pm o(1)$ according to Lemma~\ref{wkeitenc}.

So far, we did not specify the boosting factor $\alpha$. 
However, the following intuitive reasoning already shows that $\alpha$ should be bounded from above and below: If $\alpha$ is too large, we risk that $\ALGa$ packs small items with high probability, even when they are not part of the optimal packing. On the other hand, by the result of Proposition~\ref{prop:no-boost} we know that $\mathsf{ALG}_1$ cannot achieve an optimal competitive ratio.
The following theorem provides lower and upper bounds on $\alpha$ such that~$\ALGa$ is~$(1/e)$-competitive.

\begin{theorem}
	\label{theo:optimalAlg12KS}
	For $1$-$2$-knapsack, algorithm $\mathsf{ALG}_\alpha$ is $(1/e-o(1))$-competitive if and only if $ 1.400382 \lesssim \alpha \leq e/(e-1) $ and $c=1/e$, assuming $n\rightarrow\infty$.
\end{theorem}
\begin{proof}
	For any item $x \in \mathcal{I}$, let $\rho(x)$ denote the global rank of $x$ after boosting.
	On a high level, we need to consider two cases.
	
	In the first case, the optimal packing contains a single item $x$. 
	If $\rho(x) = 1$, we immediately obtain $\E [v(\ALGa)] \geq p_1 v_x = (1/e) \cdot \OPT$.	
	Now, suppose $\rho(x) \geq 2$. Let $a$ and $b$ be the items such that $\rho(a)=1$ and $\rho(b)=2$, respectively.
	Hence, 
	\[v'_a > v'_b \geq v'_x \geq \OPT \,.\]
	We note that $a$ is small, as otherwise $v_a = v'_a > \OPT$.
	Moreover, for $\alpha < 2$, item $b$ is large: If $b$ was small, it would follow that $v'_b = \alpha \cdot v_b$ and therefore~$v_a + v_b = v'_a/\alpha + v'_b/\alpha > (2/\alpha) \cdot \OPT > \OPT$, contradicting the assumption that the optimal packing contains a single item.
	Therefore, $a$ is small and $b$ is large, implying
	$v_a = v'_a/\alpha > \OPT/\alpha$ and $v_b = v'_b \geq \OPT$. Hence,
	\begin{align}
	\E[v(\ALGa)] \geq &\; p_1 \cdot v_a + p_2 \cdot v_b\label{ineq:UB-alpha}\\
	= &\; \left(\frac{1}{e}\pm o(1)\right) \cdot \frac{\OPT}{\alpha} + \left(\frac{1}{e^2}\pm o(1)\right) \cdot \OPT\nonumber\\
	\geq &\; \left(\frac{1}{e}\pm o(1)\right) \cdot \OPT \,,\nonumber
	\end{align}
	where the latter inequality holds for $\alpha \leq e/(e-1)$. Note that, when $v'_a=1$, $v'_b=1-\varepsilon$, and $v'_z=O(\varepsilon)$ for all other items $y$, Inequality~\eqref{ineq:UB-alpha} becomes satisfied with equality as $\varepsilon\to0$. Therefore, $\ALGa$ is not $(1/e-o(1))$-competitive when $\alpha > e/(e-1)$.
	
	In the remainder of the proof, we consider the case where the optimal packing contains two small items $x$ and $y$, where we assume $v_x > v_y$ without loss of generality.
	We set $j := \rho(x)$ and $k := \rho(y)$, where $1 \leq j < k$. Now, let $a_1,\ldots,a_{j-1}$ and $b_{j+1},\ldots,b_{k-1}$ denote the items appearing before $x$ and between $x$ and $y$, respectively, in the ordered sequence of boosted profits:
	\[
	v'_{a_1} > \ldots > v'_{a_{j-1}} > v'_x > v'_{b_{j+1}} > \ldots > v'_{b_{k-1}} > v'_y \,.
	\]
	We observe that neither $a$ items nor $b$ items can be small: 
	Otherwise, the profit of such an item would be strictly larger than $v_y$, 
	and as any two small items fit together, this item should be in the optimal packing instead of $y$.
	Therefore, we have $v_{a_i} = v'_{a_i} > v'_x = \alpha \cdot v_x$ for all $i\in\{1,\dots,j-1\}$ and $v_{b_i} = v'_{b_i} > v'_y = \alpha \cdot v_y$ for all $i\in\{j+1,\dots,k-1\}$.
	
	Now, we can bound the expected profit of $\ALGa$ as follows:
	\begin{align}
	\E[v(\ALGa)]
	&\geq \left(\sum_{i=1}^{j-1} P_i \cdot \alpha \cdot v_x \right) + P_j \cdot v_x + \left(\sum_{i=j+1}^{k-1} P_i \cdot \alpha \cdot v_y\right) + P_k \cdot v_y \label{eq:boosting1}\\
	&= \left(\sum_{i=1}^{j-1} p_i \cdot \alpha \cdot v_x \right) + (p_j + p_k) \cdot v_x + \left(\sum_{i=j+1}^{k-1} p_i \cdot \alpha \cdot v_y\right) + 2p_k \cdot v_y \nonumber\\	
	&= \underbrace{\left(p_j + p_k + \alpha \cdot \sum_{i=1}^{j-1} p_i \right)}_{\lambda_x} \cdot\; v_x + \underbrace{\left(2p_k + \alpha \cdot \sum_{i=j+1}^{k-1} p_i \right)}_{\lambda_y} \cdot \;v_y\,,\nonumber
	\end{align}
	where we use Corollary~\ref{Wkeitencorollary} for the first equality.
	
	If $\lambda_x < \lambda_y$ we immediately get $\lambda_x v_x + \lambda_y v_y > \lambda_x (v_x + v_y) \geq p_1 (v_x + v_y) = (1/e) \cdot \OPT$. 
	Therefore, we assume $\lambda_x \geq \lambda_y$ in the following.
	By Chebyshev's sum inequality, it holds that 
	$\lambda_x v_x + \lambda_y v_y \geq (1/2) \cdot (\lambda_x + \lambda_y) \cdot (v_x+v_y)$.
	Therefore, the competitive ratio is 
	\begin{equation}
	\label{eq:crBoosting}
	\frac{\E[v(\ALGa)]}{\OPT} \geq \frac{\lambda_x + \lambda_y}{2} 
	= \frac{1}{2} \cdot \left( (1-\alpha)\cdot p_j + 3p_k + \alpha \cdot \sum_{i=1}^{k-1} p_i \right) \,.
	\end{equation}
	If $k=2$, it follows that $j=1$ and therefore Equation~\eqref{eq:crBoosting} resolves to
	\[
	\E[v(\ALGa)] \geq \frac{1}{2} \cdot ( p_1 + 3p_2 ) \cdot \OPT > \frac{1}{e} \cdot \OPT \,,
	\]
	which holds independently of $\alpha$.
	For $k\geq 3$, $\ALGa$ is $(1/e-o(1))$-competitive by Equation~\eqref{eq:crBoosting} if
	\[
	\alpha \geq \frac{2/e - p_j -3p_k}{\sum_{i=1}^{k-1} p_i - p_j} =: \theta_{j,k}\,.
	\]
	It remains to show $\theta_{j,k} \leq 1.400382$ for all $k \geq 3$ and $j$ with $1 \leq j < k$.
	For this purpose, we first show
	\begin{equation}
	\label{eq:alphaBoundKgeq3}
	\theta_{j,k} = \frac{2/e - p_j -3p_k}{\sum_{i=1}^{k-1} p_i - p_j}
	\leq \frac{2/e - p_1 -3p_k}{\sum_{i=1}^{k-1} p_i - p_1}
	= \frac{1/e -3p_k}{\sum_{i=2}^{k-1} p_i}\pm o(1)
	~~~~~~~\text{for any $k \geq 3$.}
	\end{equation}
	Since $p_j$ is decreasing in $j$, the inequality in Equation~\eqref{eq:alphaBoundKgeq3} follows immediately if we can show
	$2/e - 3p_k > \sum_{i=1}^{k-1} p_i$ for large-enough $n$. 
	This inequality is easily verified for $k=3$, as $2/e - 3p_3 > p_1 + p_2$, for large-enough $n$.
	For $k \geq 4$, note that $p_k < p_1 - 1/3$, again for large-enough $n$, which is equivalent to $2/e - 3p_k > 1 - p_1$. 
	Using Observation~\ref{obs:sumOfPi}, we obtain
	$\sum_{i=1}^{k-1} p_i < \sum_{i=1}^{n} p_i = 1-c = 1-p_1$.
	Combining both inequalities yields Equation~\eqref{eq:alphaBoundKgeq3}.
	
	By computing the last term in Equation~\eqref{eq:alphaBoundKgeq3} for $3 \leq k \leq 10$, we obtain the upper bounds on $\theta_{j,k}$ given in Table~\ref{tab:tabelle3x}, up to additive $o(1)$ terms.
	Note that the maximum value is 1.400382.
	For $k \geq 11$, we obtain from Equation~\eqref{eq:alphaBoundKgeq3} together with $p_i \geq 0$ for all $i \geq 11$ that
	\[
	\theta_{j,k} \leq \frac{1/e -3p_k}{\sum_{i=2}^{k-1} p_i}
	\leq \frac{1/e}{\sum_{i=2}^{11-1} p_i}	
	< 1.398875\pm o(1) \,.
	\]
	
	\begin{table}[t]
		\centering
		\caption{Upper bounds on $\theta_{j,k}$ for $ 3 \leq k \leq 10 $ according to Equation~\eqref{eq:alphaBoundKgeq3}.}
		\label{tab:tabelle3x}
		\begin{tabular}{c|cccccccc}
			\toprule
			$k$ & 3 & 4 & 5 & 6 & 7 & 8 & 9 & 10 \\ 
			\midrule
			$\frac{1/e -3p_k}{\sum_{i=2}^{k-1} p_i}$ & 1.3475 & 1.3962 & \textbf{1.400382} & 1.3988 & 1.3968 & 1.3952 & 1.3941 & 1.3934  \\
			\bottomrule
		\end{tabular}
	\end{table}
		
	For the lower bound of approximately $1.400382$ on $\alpha$, first note that for $j=1$ and $k=5$, it holds indeed that
	\begin{align*}
	\theta_{1,5} &= 
	\dfrac{2/e-p_1-3p_5}{\sum_{i=1}^{5-1} p_i-p_1} = \dfrac{1/e-3p_5}{p_2+p_3+p_4} \pm o(1)\\
	&= -\dfrac{51}{16} + \dfrac{9}{4e} + \dfrac{75 - 522 e + 486 e^2}{16 - 96 e + 288 e^2 - 64 e^3}\pm o(1)
	\approx 1.400382\pm o(1)\,.
	\end{align*}
Next, note that setting $v'_x$,$v'_{b_2}$,$v'_{b_3}$,$v'_{b_4}$, and $v'_y$ all equal to $1+O(\varepsilon)$ and $v'(z)=O(\varepsilon)$ for all other items $z$ makes Inequality~\eqref{eq:boosting1} as well as Inequaltiy~\eqref{eq:crBoosting} tight as $\varepsilon\rightarrow0$. Therefore, the above arguments imply that $\alpha\geq \theta_{1,5}$ if and only if $\ALGa$ is $(1/e-o(1))$-competitive. This completes the proof.
\end{proof}

%% file: impossible.tex
In this section, we consider ordinal algorithms for 1-$B$-knapsack with $B$ large. Recall that ordinal algorithms have access to both item sizes and the relative order on item values (of previously arrived items) but not to the actual item values. We show the following theorem.

\begin{theorem}\label{thm:ordinal}
There is an ordinal $(1/(e+1)-o(1))$-competitive algorithm for the $1$-$B$-knapsack problem, and every ordinal algorithm has a competitive ratio of at most $1/(e+1)+o(1)$ for this problem.
\end{theorem}

We first discuss the lower bound, i.e., the algorithm. Note that, while the input is any combination of large and small items, the optimal solution still consists of either the single most valuable item $\opt_L$ or of a set of up to $B$ small items $\opt_S$. Our algorithm can be viewed as a linear combination of (near-)optimal algorithms $\alg_L$ and $\alg_S$ against the respective cases. In particular, $\alg_L$ is the $(1/e)$-competitive algorithm~\cite{Ferg89a} for the standard secretary problem and run with probability $e/(e+1)$; $\alg_S$ is the $(1-o(1))$-competitive algorithm for $k$-secretary by Kleinberg~\cite{Kleinberg05} and run with probability $1/(e+1)$. The competitive ratio follows by a simple case distinction. A small subtlety that we need to take care of is that these subroutines require the number of items as input. To deal with this problem, we introduce dummy items. In the following, we make this idea formal.

\begin{proof}[Proof (Algorithm)]
The algorithm $\alg_L$ treats all items as if they were large and then applies the standard secretary algorithm~\cite{Lindley1961,Dynkin1963}. For the algorithm $\alg_S$, whenever a large item arrives, we pretend that a small dummy item with value $0$ arrives. These dummy items can be accepted and take up space in the capacity constraint, but they do not contribute to the solution value. On this adapted instance, we apply an optimal algorithm for the multiple-choice secretary problem, e.g. Kleinberg~\cite{Kleinberg05} or Kesselheim et al.~\cite{KesselheimRTV18}. Clearly, for both algorithms, any solution for the respective adapted instance can be translated back to a solution with equal value for the original instance. Also, both of these algorithms are ordinal.

For every input instance, our algorithm chooses $\alg_L$ with probability $\frac{e}{e+1}$ and $\alg_S$ otherwise. To analyze the competitive ratio, distinguish two cases. If $\opt=\opt_L$, we use that the algorithm chooses $\alg_L$ with probability $e/(e+1)$ and conditioned on that achieves an expected value of $v(\opt_L)/e$~\cite{Lindley1961,Dynkin1963}, yielding an unconditional expected value of $v(\opt_L)/(e+1)$. Otherwise, i.e., if $\opt=\opt_S$, we use that $\alg_S$ is run with probability $1/(e+1)$ which achieves, as $B\rightarrow\infty$, an expected value of $(1-o(1))\cdot v(\opt_S)$, resulting in an unconditional expected value of $(1-o(1))\cdot v(\opt_S)/(e+1)$.
\end{proof}

We now discuss the upper bound, i.e., the impossibility. In our construction, there are $B$ large and $B$ small items. All items have different values, and each large item is more valuable than each small item. The adversary chooses between two ways of setting the values: The first option is to make the solution consisting of \emph{all} small items much more valuable than any single large item; the second option is to make a single large item much more valuable than any other solution.

Ideally, we would like to analyze algorithms in the following setting: In each of $n$ rounds, the algorithm is presented with both a uniformly random small and a uniformly random large item out of the items not presented thus far. Upon presentation of any such two items, the algorithm has to choose whether to select all small items from now on or to select the current large item. While the actual setting, in which all items arrive in uniformly random order, is clearly different, we show below that working with the other setting is only with a $(1\pm o(1))$-factor loss in the impossibility by reductions between our problem and an auxiliary batched-arrival model.

Assuming the latter setting, we can write a linear program similar to that of Buchbinder et al.~\cite{BuchbinderJS14}. Like in that approach, each LP solution corresponds to an algorithm and vice versa. More specifically, our LP uses two variables (rather than one) for every time step, corresponding to the probabilities that the algorithm accepts a large item or the first small item, respectively. In addition, there is a variable representing the competitive ratio, and there are two upper bounds (rather than one) on that variable, representing the two instances the adversary can choose. A feasible dual solution then yields the desired impossibility. We formalize these ideas in the following.

\begin{proof}[Proof (Impossibility)]
Consider the following two instances that are treated identically by ordinal algorithms. There are $n=2B$ items where items $i\in \{1, \dots, B\}$ are large and items $i \in \{B+1, \dots, 2B\}$ are small. In one instance, the item values are $v_i= 1+(B-i)\cdot \epsilon$ for $i\leq B$ and $v_i= 1-i\epsilon$ for $i > B$. In the other instance, the values are the same except for $v_1 = B^2$. So, for both instances, the rank of item $i$ is indeed $i$, for all $i\in\{1,\dots,2B\}$. The two optimal solutions are $\opt_L = \{1\}$ and $\opt_S = \{B+1,B+2,\dots,2B\}$. The adversary decides which of the two instances is the actual instance.

We consider the following batched-arrival setting parameterized with some constant $k$ and assume that $k$ divides $n$. The items still arrive in uniformly random order, but the algorithm does not always have to make a decision upon the arrival of an item. More specifically, for any $i\in\{1,\dots,k\}$, upon the arrival of the $(i\cdot n/ k)$-th item, the algorithm may make a decision about all items that have arrived in the current batch, i.e., after the $((i-1)\cdot n/k)$-th item. Clearly, any upper bound on the competitive ratio achievable in this setting, is also an upper bound on the competitive ratio achievable in the original setting.

Note that the expected number of items of each type, i.e., small and large, in each batch is $n/(2k)$. Let $\delta>0$ be some constant. As follows from a standard concentration (e.g., Chernoff) bound, when $n\rightarrow\infty$, the probability that the number of items from each type is between $(1-\delta)\cdot n/(2k)$ and $(1+\delta)\cdot n/(2k)$ approaches $1$. From the union bound over all batches it then follows that also the probability that the number of items of each type \emph{in each batch} is within the given range approaches $1$. We may therefore assume that this is indeed the case at an arbitrarily small loss in our impossibility.

To analyze the algorithm in the batched-arrival setting, we write a linear program similar to that of Buchbinder et al.~\cite{BuchbinderJS14}. The LP encodes a probability distribution for the decisions that an algorithm $\alg$ makes against the pair of instances. The variable $p_i$ represents the probability that the algorithm selects the best large item from the $i$-th batch. Similarly, the variable $q_i$ represents the probability that the algorithm selects all small items from both the $i$-th batch and forthcoming batches. 

Note that the algorithm may make any such decision, i.e., selecting the best largest item or starting to select small items from a batch, for at most a single batch. Hence, we obtain $q_i\leq 1-\sum_{j=1}^{i-1}p_j+q_j$ as a constraint for our LP for all~$i\in[k]$. Further, observe that we may assume that the algorithm only selects a large item when the best largest item so far is in the current batch. In batch $i$, the probability for this to happen is at most $(1+\delta)/((1+\delta)+(i-1)\cdot(1-\delta))$. As $\delta\to0$, we obtain \[p_i\leq \left(1-\sum_{j=1}^{i-1}p_j+q_j\right)\cdot\frac1i\] for all $i\in[k]$, another constraint of the LP.

\begin{figure}[t]
\scriptsize
\begin{align*}
\max &&\hspace{-.7cm}c &&&&&\min &\hspace{-.8cm}\sum_{i=1}^k (x_i + y_i)&&&&\\
s.t.\ &&\hspace{-.7cm} c &\leq \frac{1}{k}\sum_{i=1}^k (i\cdot p_i) && && s.t.\ &\hspace{-.8cm} \alpha + \beta &= 1 &&\\
&&\hspace{-.7cm} c &\leq \sum_{i=1}^k \left[\left(1-\frac{i-1}{k}\right)\cdot q_i\right] && &&& \hspace{-.8cm}i\cdot x_i + \sum_{j=i+1}^{k} (x_j + y_j) &\geq  \frac{i}{k}\cdot \alpha &&\forall\; i\in[k]\\
&&\hspace{-.7cm} i\cdot p_i&\leq 1 - \sum_{j=1}^{i-1} (p_j + q_j) && \forall\; i\in[k]&&&\hspace{-.8cm} y_i + \sum_{j=i+1}^k (x_j + y_j) &\geq \left(1-\frac{i-1}{k}\right)\cdot\beta && \forall\; i\in[k]\\
&&\hspace{-.7cm} q_i & \leq 1 - \sum_{j=1}^{i-1}(p_j + q_j)&& \forall\; i\in[k]&&&\hspace{-.8cm} x_i, y_i & \geq 0 &&\forall\; i\in[k]\\
&&\hspace{-.7cm}  p_i, q_i & \geq 0 &&\forall\; i\in[k]&&&\hspace{-.8cm} \alpha,\beta & \geq 0
\end{align*}
\caption{The primal and dual linear programs used in our proof of the upper bound in Theorem~\ref{thm:ordinal}.}
\label{fig:lps}
\end{figure}

The objective function of the LP is $c$, an upper bound on the competitive ratio of the algorithm. For each of the two instances that the adversary could choose, we write an additional constraint upper bounding $c$. If the adversary chooses the first instance and the algorithm starts selecting small items at the end of the $i$-th batch, the fraction of $v(\opt_S)$ the algorithm obtains is at most \[\frac{(k-i+1)\cdot(1+\delta)}{(k-i+1)\cdot(1+\delta)+(i-1)\cdot(1-\delta)}\xrightarrow{\delta\to0}1-\frac{i-1}{k}\,.\] Hence, as $\delta\to0$, we obtain the constraint $c \leq \sum_{i=1}^k \left(1-\frac{i-1}{k}\right)q_i$. Now consider the case that the adversary chooses the second instance. Suppose that the algorithm selects the best large item from the $i$-th batch, which is by assumption the best item that has already arrived. Since the order of large items is a uniformly random order, the probability that the chosen item is the globally best large item is the fraction of already observed large items within the whole instance, that is, at most \[\frac{i\cdot(1+\delta)}{i\cdot(1+\delta)+(k-i)\cdot(1-\delta)}\xrightarrow{\delta\to0}\frac ik\,.\] Hence, as $\delta\to0$, we get the constraint $c \leq \sum_{i=1}^k (p_i\cdot\frac{i}{k}) = \frac{1}{k}\sum_{i=1}^k i \cdot p_i$. We give both the resulting LP and its dual in Figure~\ref{fig:lps}.

We give a solution to the dual LP. Let $\tau$ be the integer number such that \[\sum_{i=\tau}^{k-1} \frac{1}{i} < 1 \leq \sum_{i=\tau-1}^{k-1}\frac{1}{i}\,.\] We set $y_i = 0$ for all $i<k$, $y_k=1/((e+1)\cdot k)$, $x_i = 0$ for $i < \tau$, and \[x_i = \frac{e}{(e+1)\cdot k}\cdot \left(1 - \sum_{j=i}^{k-1}\frac{1}{j}\right)\] for $i \geq \tau$. Further, $\alpha = e/(e+1)$ and $\beta = 1/(e+1)$. Note that this choice of $x$ is analogous to the dual solution by Buchbinder et al.~\cite{BuchbinderJS14} but scaled by a factor of~$\alpha$. 

We argue that the solution is feasible when $x$ is scaled up by a $(1+o(1))$ factor (where the Landau symbol is with respect to $k\to\infty$). Clearly, $\alpha+\beta=1$. The inequality $ix_i + \sum_{j=i+1}^{k} x_j + y_j \geq  \frac{i}{k}\cdot \alpha$ is the same as in the dual by Buchbinder et al., except for additional $y$ variables on the left-hand side and a scaling by $\alpha$ of the right-hand side. Therefore with our choice of $x$ (which is scaled up by $\alpha$ compared to Buchbinder et al.), the inequalities are identical and the previous proof of feasibility also holds, even without additional scaling of $x$. We consider the remaining (new) inequalities. For $i=1$, we have
\begin{align*}y_1 + \sum_{j=2}^k x_j + y_j &\geq \frac{e}{(e+1)\cdot k}\cdot \sum_{j=\tau}^k\left(1 - \sum_{\ell=j}^{k-1}\frac{1}{\ell}\right)\\
&= \frac{e}{(e+1)\cdot k} \cdot \left((k-\tau+1) - \sum_{j=\tau}^k\sum_{\ell=j}^{k-1}\frac{1}{\ell}\right)\\
&= \frac{e}{(e+1)\cdot k} \cdot \left((k-\tau+1) - \sum_{j=\tau+1}^{k}\frac{j-\tau}{j}\right)\\
&= \frac{e}{(e+1)\cdot k} \cdot \left(1 + \tau\sum_{j=\tau+1}^{k}\frac{1}{j}\right) \geq \frac{1}{1 + o(1)}\cdot \frac{1}{e+1}\,.
\end{align*}
For $1<i<\tau$ the corresponding inequality is weaker than the latter inequality. For $i \geq \tau$, we have
\begin{align*}y_i + \sum_{j=i+1}^k x_j + y_j = 	y_k+\sum_{j=i+1}^k x_j.
\end{align*}
Since $y_k=\beta/k$, we therefore have to show that, after scaling $x$ up by a $(1+o(1))$-factor, $\sum_{j=i+1}^k x_j$ is at least as large as $(1-i/k)\cdot\beta$. This is clear for $i=k$. For~$\tau\leq i\leq k-1$,
\begin{align*}
\sum_{j=i+1}^k x_j&= \frac{e}{(e+1)\cdot k}\cdot \sum_{j=i+1}^k\left(1 - \sum_{\ell=j}^{k-1}\frac{1}{\ell}\right)\\
&= \frac{e}{(e+1)\cdot k} \cdot \left((k-i) - \sum_{j=i+1}^k\sum_{\ell=j}^{k-1}\frac{1}{\ell}\right)\\
&= \frac{e}{(e+1)\cdot k} \cdot \left((k-i) - \sum_{j=i+1}^{k-1}\frac{j-i}{j}\right)\\
&= \frac{e}{(e+1)\cdot k} \cdot \left(1 + i\sum_{j=i+1}^{k-1}\frac{1}{j}\right) \geq \frac{1}{1 + o(1)}\cdot \left(1-\frac ik\right)\cdot\frac{1}{e+1}\,.
\end{align*}
Similar to the previous calculations, the objective-function value is
\begin{align*}
&(1+o(1))\cdot\frac{e}{(e+1)\cdot k}\cdot\sum_{j=\tau}^k\left(1-\sum_{\ell=j}^{k-1}\frac{1}{\ell}\right) \\
= \;&(1+o(1))\cdot\frac{e}{(e+1)\cdot k}\cdot \left(1 + \tau\sum_{j=\tau+1}^{k}\frac{1}{j}\right)\\
\leq \;& (1+o(1))\cdot\frac{e}{(e+1)\cdot k}\cdot \tau\\
\leq \;& (1+o(1))\cdot\frac{1}{e+1}\,,
\end{align*}
as claimed.
\end{proof}

%% file: conclusion.tex
In this paper, we have established that the $1$-$2$-knapsack secretary problem is no harder than the classic secretary problem in a competitive-ratio sense. While we previously noticed that our technique cannot directly be extended to the general setting, we believe that our work is a first non-trivial step within the larger research plan of settling the achievable competitive ratio for general knapsack secretary.

It seems plausible that our result extends to the setting of arbitrary knapsack size $B$ and item sizes $1$ or $2$. One approach may be combining our techniques with simple $1/e$-competitive algorithms for $k$-secretary~\cite{BabaioffIKK07}. More general variants seem to require handling packings of items of various sizes. A variant that avoids considering such potentially complicated configurations and may still yield an impossibility of larger than $1/e$ is $1$-$B$-knapsack.

%% file: apx.tex
\section{Proof of Propostion~\ref{prop:no-boost}}

\begin{proof}[Proof of Propostion~\ref{prop:no-boost}]
	In the first case, the optimum consists of a single element, thus, $ v(\mathsf{OPT}) = v_1 $. Since $\ALG$ chooses this element with probability $ P_1 \geq p_1 $, we have
	\begin{equation}\label{eq:firstCase1}
	\mathbb{E}[v(\mathsf{ALG})] \geq p_1 \cdot v_1 = p_1 \cdot v(\mathsf{OPT}) \geq (0.35317\pm o(1)) \cdot v(\mathsf{OPT}) \,,
	\end{equation}
	where for the last inequality we used Lemma~\ref{wkeitenc} with $c = 0.26888$. 
	
	Now, assume that the optimal packing contains two elements $x$ and $y$, where we assume $x<y$ w.l.o.g.. 
	Hence, $v(\mathsf{OPT}) = v_x + v_y$. Note that $x$ and $y$ must be the most profitable items among the set of small items, i.e., $ \mathrm{r}_\mathrm{s}(x)=1 \text{ and } \mathrm{r}_\mathrm{s}(y)=2 $.
	Next, we bound the expected profit of the packing. 
	Since $v_i \geq v_x$ for $1 \leq i \leq x$ and $v_i \geq v_y$ for $x+1 \leq i \leq y$,
	we obtain
	\[
	\mathbb{E}[v(\mathsf{ALG})] 
	\geq \sum\limits_{i=1}^y (P_i \cdot v_i)
	= v_x \cdot \sum\limits_{i=1}^{x} P_i  + v_y \cdot \sum\limits_{i=x+1}^{y} P_i \,.
	\]
	Define $\lambda_x := \sum_{i=1}^{x} P_i$ and $\lambda_y := \sum_{i=x+1}^{y} P_i$.
	If $\lambda_x < \lambda_y $, then $\lambda_y > \lambda_x \geq p_1$ 
	and thus $\mathbb{E}[v(\mathsf{ALG})] \geq v_x \cdot p_1 + v_y \cdot p_1 = p_1 \cdot v(\mathsf{OPT})$,
	which gives the same bound as in \eqref{eq:firstCase1}.
	Therefore, we assume $ \lambda_x \geq \lambda_y $ in the following.
	Since $v_x \geq v_y$, applying Chebyshev's sum inequality gives
	\begin{equation*}
	\mathbb{E}[v(\mathsf{ALG})] 
	\geq \left( \dfrac{v_x + v_y}{2}\right) \left( \lambda_x + \lambda_y \right) 
	=\frac{ v(\mathsf{OPT}) }{2}\cdot\sum\limits_{i=1}^{y} P_i 
	= \frac{v(\mathsf{OPT})}{2}\cdot \left(2p_y+\sum\limits_{i=1}^{y} p_i \right) , 
	\end{equation*}
	where the last step results from Corollary \ref{Wkeitencorollary} with $ \mathrm{r}_\mathrm{s}(x)=1 \text{ and } \mathrm{r}_\mathrm{s}(y)=2 $.
	Let $\theta_y = (1/2) \cdot \sum_{i=1}^{y} p_i +p_y$. 
	By calculating $p_i$ for $1 \leq i \leq 7$ and $c=0.26888$ using Lemma~\ref{wkeitenc}, we obtain the following upper bounds up to additive $o(1)$ terms:
	\begin{center}
		\begin{tabular}{c|cccccc}
			& $y=2$ & $y=3$ & $y=4$ & $y=5$ & $y=6$ & $y=7$ \\
			\hline
			$\theta_y$ & 0.4115 & 0.3820 & 0.3718 & 0.3678 & 0.3662 & \textbf{0.3656} \\
		\end{tabular}
	\end{center}
	Thus, for $ y \in \{2,\ldots,7\} $ we have $\theta_y \geq \theta_7$ for large-enough $n$.
	For $ y \in \{8,...,n\} $, we get $\theta_y = (1/2) \cdot \sum_{i=1}^{y} p_i +p_y
	\geq (1/2) \cdot \sum_{i=1}^{7} p_i = \theta_7 - p_7$.
	Hence, for any $y \geq 2$ it holds that $\theta_y \geq \theta_7 - p_7 \geq 0.35317\pm o(1)$.
	Overall, in the second case it holds that
	\begin{equation*}
	\begin{aligned}
	\mathbb{E}[v(\mathsf{ALG})] & \geq \frac{v(\mathsf{OPT})}{2}\cdot \left(2p_y+\sum\limits_{i=1}^{y} p_i \right) \geq (\theta_7 - p_7) v(\mathsf{OPT}) \geq (0.35317\pm o(1))\cdot v(\mathsf{OPT}) \,.
	\end{aligned}
	\end{equation*}
	This completes the proof.
\end{proof}



%% file: ms.bbl
\begin{thebibliography}{10}

\bibitem{Albers0L20}
S.~Albers, A.~Khan, and L.~Ladewig.
\newblock Best fit bin packing with random order revisited.
\newblock {\em Algorithmica}, 83(9):2833--2858, 2021.

\bibitem{AlbersKL21}
S.~Albers, A.~Khan, and L.~Ladewig.
\newblock Improved online algorithms for knapsack and {GAP} in the random order
  model.
\newblock {\em Algorithmica}, 83(6):1750--1785, 2021.

\bibitem{AlbersL21}
S.~Albers and L.~Ladewig.
\newblock New results for the \emph{k}-secretary problem.
\newblock {\em Theor. Comput. Sci.}, 863:102--119, 2021.

\bibitem{BabaioffIKK07}
M.~Babaioff, N.~Immorlica, D.~Kempe, and R.~Kleinberg.
\newblock A knapsack secretary problem with applications.
\newblock In {\em Conference on Approximation Algorithms for Combinatorial
  Optimization Problems (APPROX)}, pages 16--28, 2007.

\bibitem{BabaioffIKK18}
M.~Babaioff, N.~Immorlica, D.~Kempe, and R.~Kleinberg.
\newblock Matroid secretary problems.
\newblock {\em J. {ACM}}, 65(6):35:1--35:26, 2018.

\bibitem{BockenhauerKKR14}
H.~B{\"{o}}ckenhauer, D.~Komm, R.~Kr{\'{a}}lovic, and P.~Rossmanith.
\newblock The online knapsack problem: Advice and randomization.
\newblock {\em Theor. Comput. Sci.}, 527:61--72, 2014.

\bibitem{BoyarFL22}
J.~Boyar, L.~M. Favrholdt, and K.~S. Larsen.
\newblock Online unit profit knapsack with untrusted predictions.
\newblock In {\em Scandinavian Symposium and Workshops on Algorithm Theory
  (SWAT)}, pages 20:1--20:17, 2022.

\bibitem{BuchbinderJS14}
N.~Buchbinder, K.~Jain, and M.~Singh.
\newblock Secretary problems via linear programming.
\newblock {\em Math. Oper. Res.}, 39(1):190--206, 2014.

\bibitem{ChanCJ15}
T.~H. Chan, F.~Chen, and S.~H. Jiang.
\newblock Revealing optimal thresholds for generalized secretary problem via
  continuous {LP:} impacts on online $k$-item auction and bipartite
  $k$-matching with random arrival order.
\newblock In {\em {ACM-SIAM} Symposium on Discrete Algorithms (SODA)}, pages
  1169--1188, 2015.

\bibitem{CorreaDFS19}
J.~R. Correa, P.~D{\"{u}}tting, F.~A. Fischer, and K.~Schewior.
\newblock Prophet inequalities for independent and identically distributed
  random variables from an unknown distribution.
\newblock {\em Math. Oper. Res.}, 47(2):1287--1309, 2022.

\bibitem{CorreaDFSZ21}
J.~R. Correa, P.~D{\"{u}}tting, F.~A. Fischer, K.~Schewior, and B.~Ziliotto.
\newblock Streaming algorithms for online selection problems.
\newblock In {\em Innovations in Theoretical Computer Science (ITCS)}, pages
  86:1--86:1, 2021.

\bibitem{Dinitz13}
M.~Dinitz.
\newblock Recent advances on the matroid secretary problem.
\newblock {\em {SIGACT} News}, 44(2):126--142, 2013.

\bibitem{Dynkin1963}
E.~Dynkin.
\newblock The optimum choice of the instant for stopping a {M}arkov process.
\newblock {\em Soviet Math. Dokl.}, 4, 1963.

\bibitem{abs-2011-01559}
T.~Ezra, M.~Feldman, N.~Gravin, and Z.~G. Tang.
\newblock General graphs are easier than bipartite graphs: Tight bounds for
  secretary matching.
\newblock In {\em {ACM} Conference on Economics and Computation (EC)}, pages
  1148--1177, 2022.

\bibitem{FeldmanSZ18}
M.~Feldman, O.~Svensson, and R.~Zenklusen.
\newblock A simple \emph{O}(log log(rank))-competitive algorithm for the
  matroid secretary problem.
\newblock {\em Math. Oper. Res.}, 43(2):638--650, 2018.

\bibitem{Ferg89a}
T.~S. Ferguson.
\newblock Who solved the secretary problem?
\newblock {\em Statistical Science}, 4(3):282--289, 1989.

\bibitem{GilibertiK21}
J.~Giliberti and A.~Karrenbauer.
\newblock Improved online algorithm for fractional knapsack in the random order
  model.
\newblock In {\em Workshop on Approximation and Online Algorithms ({WAOA})},
  pages 188--205, 2021.

\bibitem{gupta_singla_2021}
A.~Gupta and S.~Singla.
\newblock Random-order models.
\newblock In T.~Roughgarden, editor, {\em Beyond the Worst-Case Analysis of
  Algorithms}, page 234–258. Cambridge University Press, 2021.

\bibitem{HoeferK17}
M.~Hoefer and B.~Kodric.
\newblock Combinatorial secretary problems with ordinal information.
\newblock In {\em International Colloquium on Automata, Languages, and
  Programming (ICALP)}, pages 133:1--133:14, 2017.

\bibitem{Kenyon96}
C.~Kenyon.
\newblock Best-fit bin-packing with random order.
\newblock In {\em {ACM-SIAM} Symposium on Discrete Algorithms (SODA)}, pages
  359--364, 1996.

\bibitem{KesselheimM20}
T.~Kesselheim and M.~Molinaro.
\newblock Knapsack secretary with bursty adversary.
\newblock In {\em International Colloquium on Automata, Languages, and
  Programming (ICALP)}, pages 72:1--72:15, 2020.

\bibitem{KesselheimRTV13}
T.~Kesselheim, K.~Radke, A.~T{\"{o}}nnis, and B.~V{\"{o}}cking.
\newblock An optimal online algorithm for weighted bipartite matching and
  extensions to combinatorial auctions.
\newblock In {\em European Symposium on Algorithms (ESA)}, pages 589--600,
  2013.

\bibitem{KesselheimRTV18}
T.~Kesselheim, K.~Radke, A.~T{\"{o}}nnis, and B.~V{\"{o}}cking.
\newblock Primal beats dual on online packing lps in the random-order model.
\newblock {\em {SIAM} J. Comput.}, 47(5):1939--1964, 2018.

\bibitem{Kleinberg05}
R.~D. Kleinberg.
\newblock A multiple-choice secretary algorithm with applications to online
  auctions.
\newblock In {\em {ACM-SIAM} Symposium on Discrete Algorithms (SODA)}, pages
  630--631, 2005.

\bibitem{Lachish14}
O.~Lachish.
\newblock ${O}(\log \log \mathrm{rank})$ competitive ratio for the matroid
  secretary problem.
\newblock In {\em {IEEE} Symposium on Foundations of Computer Science
  ({FOCS})}, pages 326--335, 2014.

\bibitem{Lindley1961}
D.~Lindley.
\newblock Dynamic programming and decision theory.
\newblock {\em Appl. Statist.}, 10, 1961.

\bibitem{Marchetti-SpaccamelaV95}
A.~Marchetti{-}Spaccamela and C.~Vercellis.
\newblock Stochastic on-line knapsack problems.
\newblock {\em Math. Program.}, 68:73--104, 1995.

\bibitem{NaoriR19}
D.~Naori and D.~Raz.
\newblock Online multidimensional packing problems in the random-order model.
\newblock In {\em International Symposium on Algorithms and Computation
  ({ISAAC})}, pages 10:1--10:15, 2019.

\bibitem{SotoTV18}
J.~A. Soto, A.~Turkieltaub, and V.~Verdugo.
\newblock Strong algorithms for the ordinal matroid secretary problem.
\newblock {\em Math. Oper. Res.}, 46(2):642--673, 2021.

\bibitem{ZhouCL08}
Y.~Zhou, D.~Chakrabarty, and R.~M. Lukose.
\newblock Budget constrained bidding in keyword auctions and online knapsack
  problems.
\newblock In {\em International Workshop on Internet and Network Economics
  (WINE)}, pages 566--576, 2008.

\end{thebibliography}
